\documentclass[11pt, twoside, a4paper]{article}
\usepackage{amsmath, amsthm, amsfonts}
\usepackage[round]{natbib}
\usepackage{fullpage}

\usepackage[kerning=true]{microtype}
\usepackage[utf8]{inputenc}
\usepackage[T1]{fontenc}

\newcommand{\R}{\mathbb{R}}

\newcommand{\E}{\mathbb{E}}

\newcommand{\bs}{\boldsymbol}
\newcommand{\sgn}{\mathrm{sgn}}
\newcommand{\diag}{\mathrm{diag}}

\newcommand{\trace}{\mathrm{trace}}
\newcommand{\cov}{\mathrm{cov}}
\newcommand{\var}{\mathrm{var}}

\newtheorem{theorem}{Theorem}
\newtheorem{definition}{Definition}
\newtheorem{assumption}{Assumption}
\newtheorem{lemma}{Lemma}

\newtheorem{proposition}{Proposition}

\title{Model Selection Consistency for Cointegrating Regressions}
\author{Eduardo F. Mendes\\Dep. of Statistics, Northwestern University}

\begin{document}
\maketitle
\begin{abstract}
We study the asymptotic properties of the adaptive Lasso in cointegration regressions in the case where all covariates are weakly exogenous. We assume the number of candidate $I(1)$ variables is sub-linear with respect to the sample size (but possibly larger) and the number of candidate $I(0)$ variables is polynomial with respect to the sample size. We show that, under classical conditions used in cointegration analysis, this estimator asymptotically chooses the correct subset of variables in the model and its asymptotic distribution is the same as the distribution of the OLS estimate given the variables in the model were known in beforehand (oracle property). We also derive an algorithm based on the local quadratic approximation and present a numerical study to show the adequacy of the method in finite samples. 
\end{abstract}

\section{Introduction}
With the increasing access to large datasets model selection has become a main issue in econometrics modeling and also in many other areas. This problem is traditionally attacked from one of the three perspectives: sequential tests, information theoretic criteria and model shrinkage. One can see that the first two are not well fitted for variable selection in higher dimensional settings and the later has not been well adapted to the problems we face in economic time series.

The sequential testing method works in a ``general-to-specific'' approach. One starts with a large model and sequentially eliminates unnecessary variables. A problem with this method is that when the number of regressors is large the performance of this method is severely compromised and multicolinearity and spurious correlation are a huge issue. The information criteria approach works by assigning weights to the models and then by minimizing some risk function among the candidate models. In a variable selection context, one wants to choose the best subset of variables, which leads to estimating approximately $10^{p/3}$ distinct models, and choose the best one according to some risk function. Clearly this method quickly becomes not feasible and alternative methods, such as greedy model selection is used instead. Greedy model selection, or sequential model selection, is not consistent and frequently choose a local minima among all models.

Another problem that model selection in high-dimension faces is that when the number of candidate variables is greater than the number of observations, estimating the model is not feasible because the parameters are not identifiable. Model shrinkage, which has been successfully used in several areas, including computer science and genomics. The idea is to shrink to zero the coefficients that do not matter in the regression leaving only the ``relevant'' ones to be estimated. One of the consequences is that only a subset of variables are actually estimated and therefore we are able to handle more variables than observations. Among shrinkage methods the Lasso, introduced by \citet{tibshirani1996}, has received much attention and several extensions have been developed, e.g. \citet{hastiezou2005}, \citet{zou2006} and \citet{yuanlin2006} among many others. 

The Lasso estimator is given by
\begin{equation}
	\hat{\theta} = \arg\min_\theta \|Y-X\theta\|_2^2 + \lambda\|\theta\|_1,
	\label{eq:lasso}
\end{equation}
where $\theta$ is a $p\times 1$ parameter vector, $Y$ is the dependent variable and $X$ is the data matrix. It can be shown that its entire regularization path can be efficiently computed \citep{efronetal2004}, can handle more covariates than observations and under some conditions can choose the correct subset of relevant variables \citep{zhaoyu2006,wainwright2006,meinshausenbuhlmann2006,meinshausenyu2009}, however it is not consistent in general and provide biased estimates for the non-zero parameters \citep{fanli2001, knightfu2000, zou2006}. \citet{zou2006} proposed a modification that has the ``oracle'' property, meaning that the estimator of the non-zero parameters have the same distribution as if we knew them beforehand. This modification led to the adaptive Lasso given by
\begin{equation}
	\hat{\theta} = \arg\min_\theta \|Y-X\theta\|_2^2 + \lambda\sum_{j=1}^p\lambda_j|\theta_j|,
	\label{eq:classic-adalasso}
\end{equation}
where the weights $\lambda_j = |\hat{\theta}^*_j|^{-\rho}$, $0<\rho \le1$, with $\hat{\theta}_j^*$ a consistent estimate of the true parameter $\theta_{0j}$.

Extension of shrinkage estimators for the case the number of candidate variables $n$ is possible much larger than the sample size often require the ``partial ortogonality condition'' which states that the variables that do not enter in the model are only weakly correlated with the variables that enter in the model \citep{huangetal2008,huangetal2009}, or the ``Irrepresentable Condition'' which states that the coefficients of the linear regression of the variables that enter the model onto the variable that do not enter the model is bounded by $1$ \citep{zou2006,zhaoyu2006,meinshausenbuhlmann2006}.

Despite all these effort in understanding and adapting the Lasso to distinct cases, most advances are only valid for the classical i.i.d. regression framework, most often with fixed design. Little or effort has been given to time series or weekly dependent case, which is the prevalent case in economic series. \citet{wangetal2007} use a Lasso-based method to choose the autoregressive order of a regression; \citet{hsuetal2008} apply the Lasso method to choose the variables in a vector autoregressive models; \citet{caner2009} applies the Lasso method to choose variables in a weakly dependent GMM framework; \citet{canerknight2008} use a bridge estimator to find the integration order of a vector; and \citet{liaophillips2010} for selecting variables and order of integration in an error correction models. All those papers suffers from the same drawback that is the number of candidate variables (or respectivelly the total numbe of parameters for the vector case) have to be smaller than the sample size. \citet{songbickel2011} provide new results allowing the number of variables to increase with the sample size and be possibly larger than it. Such techniques have also been used in applied research in more general frameworks. For instance, \citet{baing2008} use Lasso-related techniques for factor forecasting, but since prediction is their ultimate goal (as opposed to variable selection), what matters is how ordered predictors affect the forecasts as opposed to how you choose the variables.

In this paper we discuss an extension of the adaptive Lasso to a (possibly) cointegrated regression with explanatory stationary variables, and show model selection consistency and oracle property for the method. We allow the model to select both the stationary and non-stationary variables in the regression. One problem in extending Lasso to cointegrated regressions is that the $I(1)$ and $I(0)$ parameters converge at distinct rates. We overcome this problem by setting regularization parameter for the $I(1)$ variables to be proportional to the square of the $\lambda$ for the stationary variables. We also relax the need of a ``zero-consistent estimator`` in \citet{huangetal2008}, imposing a weaker form of the ''Irrepresentable Condition``.

Throughout the paper we assume it is already known the order of integration of the dependent and independent variables. We consider the case where the actual number of $I(1)$ variables in the model, $q_1$, is fixed, but the number of $I(0)$ variables in the model, $q_2$, can increase with $T$. Moreover, the total number of candidate $I(1)$ variables is sub-linear with respect to the sample size $T$, meaning that the number of candidate variables $n_1$ is $o(T)$, but possibly larger than $T$. This last condition can be relaxed if more structure is imposed on the error term of the regression, and we can achieve a rate for $n$ as big as $o(e^{T^\delta})$, for some $0\le\delta\le1$ \citep{huangetal2008}. Similarly, the number of candidate $I(0)$ variables, $n_2$, is $o(T^d)$, for some $d\ge 1$. The results in this paper can also be extended to the (finite) vector-case and also (independent) panel data models.

One of the most straightforward application of this result is to understand the shift in prices of financial objects (financial portfolio construction). The prices are known to be $I(1)$ and number of financial objects that might of interest is large and include both $I(1)$ and $I(0)$ variables. Another interesting framework is the evolution of macroeconomic time series, as in \citet{stockwatson2002}. The number of predictors can be very large and an efficient method for choosing the relevant ones is necessary. Another application of this method is to choose the number of lags in a Autoregressive Distributed Lags (ADL) model.

In section \ref{model} we present the proposed model selection method. Section \ref{asymp} presents the main results of the paper. Section \ref{algo-sim} shows the algorithm for estimating the parameters and a Monte Carlo study to evaluate the performance of the method in finite samples. We close the paper with some final remarks in section \ref{conclusion}. The proof of the main results are delayed to the appendix.

\section{Penalized Cointegration}\label{model}

Let $\{y_t\}_1^\infty$ denote an scalar time series generated by
\begin{equation}
	y_t  = \alpha_0 + \beta_0'x_t + \gamma_0' z_t + u_t
	\label{eq:modely}
\end{equation}
where $\alpha_0$ is a scalar, $\beta_0$ is $n_1\times 1$, and $\gamma_0$ is $n_2\times1$, with the index $\cdot_0$ meaning ``true''. The process $\{x_t\}_1^\infty$ satisfies
\begin{equation}
	x_t = x_{t-1} + v_t,
	\label{eq:modelx}
\end{equation}
the process $\{z_t\}_1^\infty$ has mean zero and is weakly stationary, and $\{u_t\}_1^\infty$ and $\{v_t\}_1^\infty$ are weakly stationary error processes. Also, the following assumption hold for the vector $w_t = (u_t, v_t', z_t)'$
\begin{assumption}[DGP]\label{a:ip}
	The vector process $\{w_t\}_1^\infty$ satisfy the following assumptions
	\begin{enumerate}
		\item $\E w_t = 0$ for $t=1,2,\dots$;
		\item $\{w_t\}_1^\infty$ is weakly stationary;
		\item for some $d>1$
			\begin{itemize}
				\item $\E|w_t|^{2d} < \infty$ for $t=1,2,\dots$; and
				\item the process $\{w_t\}_1^\infty$ is either $\phi$-mixing with rate $1-1/(2d)$, or $\alpha$-mixing with rate $1-1/d$.
			\end{itemize}
		\item The process $\{u_t\}_1^t$ is uncorrelated with $\{v_t\}_1^t$ and $\{z_t\}_1^t$, for $t=1,2,\dots$
		\item Define $S_t = \sum_1^T w_tw_t'$. Then
			\[
			\begin{array}{ll}
				\lim_{T\rightarrow\infty} T^{-1}\E S_tS_t' &= \E w_1w_1' + \sum_{t=1}^\infty\E[w_1w_t'+w_tw_1']\\ 
				&= \Sigma + \Lambda + \Lambda'\\
				&= \Sigma^*.
			\end{array}.
			\]
		\item $\max_{j=1,\dots,n_2}\left(\E\left(|T^{1/2}\sum_{t=1}^Tz_{jt}u_t|\right)^{2d}\right)^{1/d}\le c_d < \infty$
		\item if $q_2\rightarrow\infty$, $\max_{1\le i\le j \le q_2} \E\left(T^{-1/2}\sum_{t=1}^tz_{it}z_{jt}-\E(z_{it}z_{jt})\right)^2 \le c_s <\infty$
		\item the eigenvalues of the matriz $\Sigma^*_{Z(1)^2}$ (the part of $\Sigma^*$ corresponding to the variables $z$ that enter in the model) are bounded between $\tau_*$ and $\tau^*$.
	\end{enumerate}
\end{assumption}

The set of assumptions (1)--(5) is common in cointegration regression. Assumptions (6) and (7) are required to control the number of $I(0)$ variables in the model. In particular, \citet{phillipsdurlauf1986} make de same set of assumptions (1-5) to derive asymptotic properties of multiple regressions with integrated processes. This assumption is required to ensure that the Invariance Principle holds. A weaker set of assumptions, using mixingales, could be used instead \citep{dejongdavidson2000}, but we decided to use the classical set of assumptions for sake of simplicity (and clarity since these are the most commonly used). The number of finite moments $d$ is directly related to the order of increase of candidate variables in the model.

In this work, we assume that $n\equiv n_T=n1+n2$ is possibly greater than $T$, but only a fraction of these coefficients are in fact nonzero. Without any loss of generality we assume each coefficient vectors can be partitioned into zero and non-zero coefficients, i.e. $\beta_0 = (\beta_0(1)',\beta_0(2)')'$ and $\gamma_0 = (\gamma_0(1)',\gamma_0(2)')'$, with all non-zero coefficients stacked first, where $\beta_0(1)$ is $q_1\times 1$, and $\gamma_0(1)$ is $q_2\times1$. We assume $q_1$ is fixed (do nor depend on $T$) and $q_2$ may depend on T, also set $q=q_1+q_2$. For matter of convenience, denote $m_1 = n_1-q_2$ and $m_2=n_2-q_2$. Denote by upper case letters the data matrices and allow splitting these matrices in the same way we did with the coefficients, for instance $Z = (z_1,\dots,z_T)' = (Z(1),Z(2))$ and $X = (x_1,\dots,x_T)' = (X(1),X(2))$. 

The Adaptive Lasso estimate in our case is given by
\begin{equation}
	(\hat\beta,\hat\gamma) = \arg\min_{\beta,\gamma} \|Y-X\beta-Z\gamma\|_2^2 + \lambda_1\sum_{j=1}^{n_1}\lambda_{1j}|\beta_j| + \lambda_2\sum_{j=1}^{n_2}\lambda_{2j}|\gamma_j|,
	\label{eq:adalasso}
\end{equation}
where $\{\lambda_1,\lambda_{11},\dots,\lambda_{1n_1},\lambda_2,\lambda_{21},\dots,\lambda_{2n_2}\}$ are regularization parameters satisfying a set of conditions defined later, and $\|\cdot\|_2^2$ denote the $L_2$-vector norm. Following \citet{zou2006}, we take $\lambda_{1j} = |\hat\beta^*_j|^{-\rho}$ and $\lambda_{2j} = |\hat\gamma^*_j|^{-\rho}$, where $\hat\beta^*_j$ and $\hat{\gamma}^*_j$ are estimators of $\beta_{0j}$ and $\gamma_{0j}$; and $0\le\rho<1$.

We assume without loss of generality that the true intercept $\alpha_0=0$ is known. This assumption does not change our results since we are interested in the behavior of the selection procedure. We make the following regularity assumptions about the parameter space $\Theta_n$ and the true vector of parameters $\theta_0 = (\beta_0', \gamma_0')'$.
\begin{assumption}\label{a:param}
	(i) The true parameter vector $\theta_0$ is an element of an open subset $\Theta_n\subset \R^{n}$ that contains the element $\bs{0}$. (ii) $\min\beta_0(1)\ge\beta_*$ and $\min\gamma_0(1))\ge\gamma_*$.
\end{assumption}

The minimization problem in (\ref{eq:adalasso}) is equivalent to a constrained concave minimization problem, and necessary and (almost) sufficient conditions \citep{zhaoyu2006} for the existence of a solutions can be derived satisfying the Karush-Kuhn-Tucker (KKT) conditions. This approach has been applied in several papers including \citet{wainwright2006}, \citet{zhaoyu2006}, \citet{zou2006} and \citet{huangetal2008}, and lead to a necessary condition frequently denote in the literature by \textit{Irrepresentable Condition} (IC). This condition is know to be easily violated in the presence of highly correlated variables \citep{zhaoyu2006,meinshausenyu2009}. \citet{meinshausenyu2009} examine the performance of the Lasso estimate in the case this condition is violated. A more comprehensive discussion about the IC and comparison with other conditions can be found in \citet{zhaoyu2006} and \citet{meinshausenyu2009}, section 1.5. 

In opposition to \citet{zou2006} and \citet{huangetal2008}, who assume one has consistent zero-estimators of the parameters $\theta_{0}(2)$, we do not assume such estimators are available; instead, we assume a weaker form of the Irrepresentability Condition denoted \textit{Weak Irrepresentability Condition} (WIC). This condition reduces to the IC if we have $P\left(\min_{ q_1+1\le j \le n_1}\lambda_{1j}=|\beta_*|^{-1}\right)\rightarrow1$ and $P\left(\min_{q_2+1\le j\le n_2}\lambda_{2j} = |\gamma_*|^{-1}\right)\rightarrow1$; and is equivalent to zero-consistency if $\lambda_{1j}$ and $\lambda_{2j}$ diverge as $T$ increase. One should expect to be in the between most of the time rendering this condition less restrictive than both IC and zero-consistency. \textit{Weak Irrepresentability Condition} also implies that we do not need consistent estimators of $\theta_0(2)$ anymore to construct $\lambda_{ij}$, $i=1,2$ and $j=q_i+1,\dots,n_i$, rather we can use biased estimators such as ridge estimators.

\begin{lemma}[KKT Conditions]
	The solution $\hat{\beta}=(\hat{\beta}(1)',\hat{\beta}(2)')'$ and $\hat{\gamma} =(\hat{\gamma}(1)',\hat{\gamma}(2)')'$ to the minimization problem (\ref{eq:adalasso}) exists if:
	\begin{subequations}\label{eq:KKT1}
		\begin{align}
			\frac{\partial \|Y-X\beta-Z\gamma\|_2^2}{\partial \beta_j(1)}\Big\vert_{\beta_j(1)=\hat{\beta}_j(1)} = \sgn(\hat{\beta}_j(1))\lambda_1\lambda_{1j} \label{eq:KKT1a}\\
			\frac{\partial \|Y-X\beta-Z\gamma\|_2^2}{\partial \gamma_j(1)}\Big\vert_{\gamma_j(1)=\hat{\gamma}_j(1)} = \sgn(\hat{\gamma}_j(1))\lambda_2\lambda_{2j} \label{eq:KKT1b}
		\end{align}
	\end{subequations}
	and
	\begin{subequations}\label{eq:KKT2}
		\begin{align}
			\frac{\partial \|Y-X\beta-Z\gamma\|_2^2}{\partial \beta_j(2)}\Big\vert_{\beta_j(2)=\hat{\beta}_j(2)}\leq \lambda_1\lambda_{ij} \label{eq:KKT2a}\\
			\frac{\partial \|Y-X\beta-Z\gamma\|_2^2}{\partial \gamma_j(2)}\Big\vert_{\gamma_j(2)=\hat{\gamma}_j(2)} \leq \lambda_2\lambda_{2j} \label{eq:KKT2b}.
		\end{align}
	\end{subequations}

	\label{l:KKT}
\end{lemma}
\begin{proof}
The proof of this lemma is simply the statement of the KKT conditions adapted to our problem.
\end{proof}

Following \citet{zhaoyu2006}, model selection consistency is equivalent do \textit{sign consistency}. We say that $\hat{\theta}$ equals in sign to $\theta$ if $\sgn(\hat{\theta}) = \sgn(\theta)$, and we represent this equality of signs by $\hat{\theta} =_s \theta$.
\begin{definition}[Sign Consistency]
	We say that an estimate $\hat{\theta}$ is sign consistent to $\theta$ if
	\[
	\Pr(\hat{\theta} =_s \theta) \rightarrow 1\quad\mbox{, as } n\rightarrow \infty.
	\]
\end{definition}

\citet{zhaoyu2006} refer to this kind of consistency as \textit{strong sign consistency}, meaning that one can use a pre-selected regularization parameter to achieve sign consistency, as opposed to \textit{general sign consistency} which states that for a random realization there exists a amount of regularization that selects the true model. 

Before stating the IC to our problem, we have to introduce some more notation. Let $W(1) = (X(1), Z(1))$, $W(2) = (X(2), Z(2))$ and $W = (W(1)W(2))$, then $\Omega = \Gamma^{-1/2}W'W\Gamma^{-1/2}$ can be divided into four blocks, $\Omega_{11} = \Gamma_1^{-1/2}W(1)'W(1)\Gamma_1^{-1/2}$, $\Omega_{21} = \Gamma_2^{-1/2}W(2)'W(1)\Gamma^{-1/2}_2$, $\Omega_{12}$ and $\Omega_{22}$. The normalization matrix $\Gamma$, is also divided in $\Gamma^{1/2}_1 = \diag(T\bs{1}_{q_1}',\sqrt{T}\bs{1}_{q_2})$ and $\Gamma^{1/2}_2=\diag(T\bs{1}_{n_1-q_1}',\sqrt{T}\bs{1}_{n_2-q_2}')$ are the following

\begin{assumption}[Weak Irrepresentable Condition]\label{a:ic}
	The matrix $\Omega_{11}$ is invertible, and for some $0<\eta<1$,  
	\[
	P\left(\bigcap_{1\le j\le m_1}\left\{\left[\big|[\Omega_{21}\Omega_{11}^{-1}]\sgn(\theta_0(1))\big|\right]_j \le \beta_*\lambda_{1j}-\eta\right\}\right)\rightarrow 1,
	\]
	and
	\[
	P\left(\bigcap_{m_1+1\le j\le m_1+m_2}\left\{\left[\big|[\Omega_{21}\Omega_{11}^{-1}]\sgn(\theta_0(1))\big|\right]_j \le \gamma_*\lambda_{2j}-\eta\right\}\right)\rightarrow 1,
	\]
	where $[\cdot]_j$ denotes the $j^{th}$ element of the vector inside brackets.
\end{assumption}

Next proposition (similar to proposition 1 in \citet{huangetal2008}) provides some lower bounds on the probability of Adaptive Lasso choosing the correct model.

\begin{proposition} Let $\lambda = \diag(\lambda_1\bs{1}_{h_1},\lambda_2\bs{1}_{h_2})$, where the dimensions $h_1$ and $h_2$ are adapted to each case it appears,  $L(1) = \diag(\lambda_{11}, \dots, \lambda_{1q_1}, \lambda_{21}, \dots,\lambda_{2q_2})$ and $L(2) = \diag(\lambda_{1q_1+1}, \dots, \lambda_{1n_1}, \lambda_{2q_2+1}, \dots, \lambda_{2n_2})$. Then

	\[
	\Pr\left( \hat\theta=_s \theta_0\right) \ge \Pr\left( \mathcal{A}_T\cap\mathcal{B}_t \right),
	\]
where

	\begin{subequations}\label{eq:necessary}
		\begin{align}
			\mathcal{A}_T &= \left\{ \Gamma^{-1/2}|\Omega_{11}^{-1}W(1)'U| < \Gamma^{1/2}|\theta_0(1)|-\frac{1}{2}\Gamma^{-1/2}\lambda|\Omega_{11}^{-1}L(1)\sgn(\theta_0(1))| \right\}\\
			\mathcal{B}_T &= \left\{ 2|\Gamma^{-1/2}W(2)'M(1)U| < \Gamma^{-1/2}\lambda L(2)\bs{1}_{n-q}-\Gamma^{-1/2}\lambda|\Omega_{21}\Omega_{11}^{-1}L(1)\sgn(\theta_0(1))|\right\},
		\end{align}
	\end{subequations}
	where $M(1) =\bs{I}_T-W(1)(W(1)'W(1))^{-1}W(1)'$ and the previous inequalities hold element-wise.
	\label{p:necessary}
\end{proposition}

\section{Model Selection Consistency and Oracle Property}\label{asymp}

In this section we derive the main results of the paper. We show that, under some conditions on $n$, $p$, and $\lambda$'s the Adaptive Lasso selects the correct subset of variables (sign consistency) and it has the oracle property in the sense of \citet{fanli2001}, meaning that our estimate has the same asymptotic distribution of the OLS as if we knew beforehand what variables are in the model and at optimal rate. A straightforward conclusion is that we can carry out hypothesis tests about the parameters in a traditional way, i.e. as if we assume we have the true model.

In our case, the number of variables $q=q_1+q_2$ that actually enter in the model can grow polynomially with $T$, more precisely the number of $I(1)$ variables $q_1$ in the model is finite while the number of $I(0)$ variables in the model can increase polynomially. The number of candidate variables $n=n_1+n_2$ increase with $T$ (both $n_1$ and $n_2$ increase with $T$ at distinct rates) and is possibly larger than the sample size. The next assumption give sufficient conditions for model selection consistency.
\begin{assumption}\label{a:lasso}
The follow assumptions hold jointly for some fixed $0<\rho\le 1$ :
	\begin{enumerate}
		\item $\lambda_1 \rightarrow \infty$ and $\lambda_1/T^{1+\rho} \rightarrow 0$;
		\item $\lambda_2 \rightarrow \infty$ and $\lambda_2/T^{(1+\rho)/2}\rightarrow 0$;
		\item $q_1=O(1)$ and $q_2=o(T^{d/(2d+1)})$;
		\item $m_1=o(T^2/\lambda_1^2)$ and $m_2=o(T^d/\lambda_2^2))$.
	\end{enumerate} 
\end{assumption}

This assumption tells us that the number of variables is sub-linear with respect to the sample size $T$, however this assumption can be relaxed at a cost of more structure about the tails of the error term. 

\begin{assumption} \label{a:adalasso}
The following assumptions hold jointly for some fixed $0<\rho\le 1$:
\begin{enumerate}
	\item There exist constants $\beta_*$ and $\gamma_*$ such that:
		\begin{enumerate}
			\item[(i)] $\Pr(\max_{1\le j \le q_1}\lambda_{1j}<\beta_*^{-1}) \rightarrow 0$; 
			\item[(ii)] $\Pr(\max_{1\le j \le q_2}\lambda_{2j}<\gamma_*^{-1}) \rightarrow 0$;
		\end{enumerate}
	\item There exists stationary processes $V_{1j}$, $j=1,\dots,q_1$, and $V_{2j}$, $j=1,\dots,q_2$ such that:
		\begin{enumerate}
			\item[(i)] $T^{\rho}\lambda_{1j} \Rightarrow V_{1j}$;
			\item[(ii)] $T^{\rho/2}\lambda_{1j} \Rightarrow V_{2j}$.
		\end{enumerate}
\end{enumerate}
\end{assumption}

The first assumption requires the weights $\lambda_{1}(1)$ and $\lambda_2(1)$ to be bounded from below with probability tending to 1. The last assumption is required for the oracle property and tells us that the data dependent weights ave to converge at a given rate for the adaptive Lasso to be oracle.

\begin{theorem}[Model Selection Consistency]
	Under assumptions \ref{a:ip} -- \ref{a:adalasso}, 
	
	\[
	P(\hat{\theta}=_s\theta_0) \rightarrow 1.
	\]
	\label{thm:modsel}
\end{theorem}

\begin{theorem}[Oracle Property]
	Suppose assumptions \ref{a:ip} to \ref{a:adalasso} are satisfied, and also that $(\lambda_2q_2)/T^{(1+\rho)/2}\rightarrow0$. Then the following holds
	\begin{equation}
		\left(
		\begin{array}{c}
			T(\hat\beta(1)-\beta_0(1))\\
			\sqrt{T}(\hat\gamma(1)-\gamma_0(1))
		\end{array}
		\right) \Rightarrow
		\left(
		\begin{array}{cc}
			\int B_{X(1)}B_{X(1)}' &  \bs{0}\\
			\bs{0}' & \Sigma_{zz}
		\end{array}
		\right)^{-1}\times
		\left(
		\begin{array}{c}
			\int_0^1B_{X(1)}dB_u\\
			N(\bs{0},\sigma^*_{u^2}\Sigma_{Z(1)^2}^*)
		\end{array}
		\right).
		\label{eq:oracle}
	\end{equation}
	\label{thm:oracle}
\end{theorem}

\section{Numerical Results}\label{algo-sim}
\subsection{Algorithm}
Since we are dealing with both $I(1)$ and $I(0)$ series, we cannot apply the plain vanilla LARS algorithm \citep{efronetal2004} to our problem, instead we will follow \citet{fanli2001} and \citet{hunterli2005} and apply a locally quadratic approximation (LQA) to the penalty function, more precisely the perturbed version in section 3.2 of \citet{hunterli2005}. This approach also allow us to derive a closed form formula for the standard error of the parameter estimates.

For a nonzero $\beta_j$ the perturbed LQA of the Adaptive Lasso penalty is given by
\begin{equation}
	\lambda_{1j}|\beta_j| \approx \lambda_{1j}|\beta_{0j}| + \frac{\lambda_{1j}}{2(|\beta_{0j}|+\varepsilon)}(\beta_j^2 - \beta_{0j}^2),
	\label{eq:lqa}
\end{equation}
for some small $\varepsilon>0$, and similarly for $\gamma_j$'s. Denote this approximation by $\psi_j(\beta_j)$; instead of minimizing \eqref{eq:adalasso}, we minimize 
\begin{equation}
	\|Y-X\beta-Z\gamma\|_2^2 +\lambda_1\sum_{j=1}^{n_1}\psi_j(\beta_j) +\lambda_2\sum_{j=1}^{n_2}\psi_j(\gamma_j)
	\label{eq:critlqa}
\end{equation}
iteratively until the estimates converge.

Define the diagonal matrix 
\[
E_k = \diag\left( \frac{\lambda_1\lambda_{11}}{(|\beta_1^{(k)}|+\varepsilon)},\dots,\frac{\lambda_1\lambda_{1n_1}}{(|\beta_{n_1}^{(k)}|+\varepsilon)},\frac{\lambda_2\lambda_{21}}{(|\gamma_1^{(k)}|+\varepsilon)},\dots,\frac{\lambda_2\lambda_{2n_2}}{(|\gamma_{n_2}^{(k)}|+\varepsilon)} \right).
\]
The estimator of $\theta^{(k+1)}$ is given by
\begin{equation}
	\theta^{(k+1)} = \left( W'W+E_k \right)^{-1}W'Y.
	\label{eq:lqaestimate}
\end{equation}

One issue with the adaptive Lasso is to find the weights $\lambda_{1j}$ and $\lambda_{2j}$. We propose to use an iterated adaptive Lasso, which consists in recalculating the weights $\lambda_{1j}$ and $\lambda_{2j}$ each step. More precisely, 
\begin{equation}
E_k = \diag\left( \frac{\lambda_1\lambda_{11}^{(k)}}{(|\beta_1^{(k)}|+\varepsilon)},\dots,\frac{\lambda_1\lambda_{1n_1}^{(k)}}{(|\beta_{n_1}^{(k)}|+\varepsilon)},\frac{\lambda_2\lambda_{21}^{(k)}}{(|\gamma_1^{(k)}|+\varepsilon)},\dots,\frac{\lambda_2\lambda_{2n_2}^{(k)}}{(|\gamma_{n_2}^{(k)}|+\varepsilon)} \right),
\label{eq:Ek}
\end{equation}
with
\begin{equation}
	\lambda_{1j}^{(k)} = |\beta_j^{(k-1)}|^{-\rho} \quad\mbox{ and }\quad \lambda_{2j}^{(k)} = |\gamma_j^{(k-1)}|^{-\rho}
	\label{eq:weightk}
\end{equation}
and the initial weights we calculate by using ridge regression with regularization parameter $\lambda^{(ridge)}$, i.e.
\begin{equation}
	\theta^{(0)} = (W'W+\lambda^{(ridge)}\bs{I}_n)^{-1}W'Y,
	\label{eq:ridge}
\end{equation}
for the best choice of $\lambda^{(ridge)}$.

This algorithm has shown to be stable in a number of simulations, with only a small change to ensure the numbers are within the margins of machine precision.

\subsection{Standard Error Formula}

\citet{hunterli2005} provide a sandwich formula for computing the covariance matrix of the penalized estimates of the nonzero components that has been proven to be consistent (\citet{fanpeng2004}). \citet{zou2006} adapted this formula to the adaptive Lasso case and is given by
\begin{equation}
	\widehat{\cov}(\hat{\theta}(1)) = \sigma^*_{uu}(W(1)'W(1)+E_k(1))^{-1}W(1)'W(1)(W(1)'W(1)+E_k(1))^{-1}.
	\label{eq:covest}
\end{equation}

If the parameter $\sigma^*_{uu}$ is unknown, one can replace it by its estimate from the full model. For the zero-valued variables, the standard errors are zero \citep{fanli2001}.

Although the consistency result derived by \citet{fanpeng2004} cannot be directly applied to our case, the same conclusion can be reached by adapting their proof to the integrated case.

\subsection{Choosing the regularization parameters}
To implement the algorithm described above, we need to estimate $\lambda_1$, $\lambda_2$ and $\lambda^{(ridge)}$. We will use the method called generalized cross-validation (GCV). 

Define the projection matrix of the ridge estimator \eqref{eq:ridge} as
\begin{equation}
	P_r({\theta(\lambda^{*})}) = W'(W'W+\lambda^{*}\bs{I}_n)^{-1}W'.
	\label{eq:proj-ridge}
\end{equation}
Hence, the number of effective parameters $\bs{e}(\lambda^*)=\trace(P_r({\theta(\lambda^*)}))$. Therefore, the GCV statistic for this problem is
\begin{equation}
	GCV_r(\lambda^{*}) = T^{-1}\frac{\|Y-W\theta(\lambda^*)\|_2^2}{(1-\bs{e}(\lambda^*)/T)^2},
	\label{eq:GCV-ridge}
\end{equation}
where $\theta(\lambda^*) = (W'W+\lambda^{^*}\bs{I}_n)^{-1}W'Y$. We find $\lambda^{(ridge)} = \arg\min_{\lambda^*}GCV_r(\lambda^*)$.

For the adaptive Lasso, define $\bs\lambda^* = (\lambda^*_1,\lambda^*_2)$ and
\begin{equation}
	E_{\bs\lambda^*} = \diag\left( \frac{\lambda^*_1\lambda_{11}}{(|\beta_1^{(0)}|+\varepsilon)},\dots,\frac{\lambda^*_1\lambda_{1n_1}}{(|\beta_{n_1}^{(0)}|+\varepsilon)},\frac{\lambda^*_2\lambda_{21}}{(|\gamma_1^{(0)}|+\varepsilon)},\dots,\frac{\lambda^*_2\lambda_{2n_2}}{(|\gamma_{n_2}^{(0)}|+\varepsilon)} \right).
	\label{eq:E-GCV}
\end{equation}
with
\begin{equation}
\lambda_{1j} = |\beta_j^{(0)}|^{-\rho} \quad\mbox{ and }\quad \lambda_{2j} = |\gamma_j^{(0)}|^{-\rho},
\label{eq:weight-GCV}
\end{equation}
where $\beta^{(0)}$ and $\gamma^{(0)}$ were estimated using \eqref{eq:ridge}. Define the projection matrix
\begin{equation}
	P_l({\theta(\bs\lambda^{*})}) = W'(W'W+E_{\lambda^*})^{-1}W'.  
	\label{eq:proj-lasso}
\end{equation}
The number of effective parameters $\bs{e}(\bs{\lambda}^*)$ is given by $\trace(P_l(\theta(\bs\lambda^*)))$, and the GCV statistic is
\begin{equation}
	GCV_l(\bs\lambda^{*}) = T^{-1}\frac{\|Y-W\theta(\bs\lambda^*)\|_2^2}{(1-\bs{e}(\bs\lambda^*)/T)^2},
\label{eq:GCV-lasso}
\end{equation}
where $\theta(\bs\lambda^*) = (W'W+E_{\bs\lambda^*})^{-1}W'Y$. We find $\bs\lambda = \arg\min_{\bs\lambda^*}GCV_l(\bs\lambda^*)$.

We perform both minimizations by doing a grid search before starting the adaptive Lasso estimation procedure. We can also include $\rho$ in the minimization of \eqref{eq:GCV-lasso}, but we found little impact between choosing $\rho$ dynamically and using it fixed at $0.9$. Smaller values for $\rho$ did affect the performance of the estimates.

\subsection{Simulation Studies}

In this section we report the results of the simulations studies. We want to evaluate the (i) model selection accuracy; (ii) estimation accuracy; and (iii) forecasting accuracy. We will consider four distinct model specifications. Each covariate is generate from a multivariate normal distribution with variance 1 and covariance structure defined in each model. We simulate each model $500$ times for three distinct sample sizes $T=50,\,100,\,200$ and an extra $50$ observations are used for evaluating prediction performance. 

\textbf{Model 1}: $u_t\sim N(0,1.5^2)$, $n_1=n_2=15$. Set $w_t = (v_t, z_t)$. The pairwise covariance between the $i$th and $j$th element of $w_t$ is given by $\cov(w_{it},w_{jt}) = r^{|i-j|}$, $r=0.5$, and $\var(w_j)=1$. The parameters $\gamma=\beta=(2.5,2.5,1.5,1.5,0.5,0.5,0,\dots,0)'$, meaning we have two large effects, two moderate effects and two weak effects for $X$ and $Z$.

\textbf{Model 2}: Similar to model 1, except that $r=0.9$.

\textbf{Model 3}: Similar to model 1, but the error term $u_t=0.6u_{t-1}+e_t$, with $e_t\sim N(0,1.5^2)$.

\textbf{Model 4}: Similar to model 3, but $n_1=n_2=50$

\textbf{Model 5}: Similar to model 1, but $n_1=n_2=50$, the first $15$ variables in $z_t$ and $u_t$ have the same dependence structure as in model 1, the remaining $2\times 35$ variables are independent.

\textbf{Model 6}: Similar to model 3, but $e_t\sim t_4$

In all examples we consider small, moderate and large effects for both $I(1)$ and $I(0)$ covariates. In model 1 we study a simple framework with a moderate number of candidate variables and weak to moderate correlation among them. In model 2 we consider the case in which the variables are highly correlated. Model 3 consider the case in which the errors have an AR(1) structure. Models 4 and 5 consider the case in which we have many variables with distinct correlations; and model 6  we consider AR(1) errors with fat tails.

\subsubsection{Model Selection Accuracy:} 
We evaluate model selection by calculating the number of corrected selected ``non-zero'' coefficients and the number of corrected selected ``zero'' coefficients. We use resampling to estimate the mean and standard deviation of the number of correct selected coefficients. In models 1, 2, 3 and 6, the number of ``zero'' coefficients is $18$; for models 4 and 5, the number of ``zero'' coefficients is $88$. For all models the number of ``non-zero'' coefficients is $12$.

\begin{table}[h]
	\caption{Variable Selection Performance}
\begin{center}
\begin{tabular}{lcccccc}
	\hline
	&\multicolumn{2}{c}{50}&\multicolumn{2}{c}{100}&\multicolumn{2}{c}{200}\\
Model &\#nz &\#z &\#nz &\#z &\#nz &\#z\\
\hline
1 & $\underset{(0.824)}{10.573}$ & $\underset{(1.367)}{16.308}$ & $\underset{(0.528)}{11.644}$ & $\underset{(1.177)}{16.860}$ & $\underset{(0.225)}{11.946}$ & $\underset{(0.837)}{17.262}$\\
2 & $\underset{(1.014)}{8.630}$ & $\underset{(1.453)}{16.605}$ & $\underset{(0.802)}{10.013}$ & $\underset{(1.034)}{17.008}$ & $\underset{(0.567)}{11.038}$ & $\underset{(0.859)}{17.320}$\\
3 & $\underset{(0.850)}{10.561}$ & $\underset{(1.485)}{15.749}$ & $\underset{(0.673)}{11.420}$ & $\underset{(1.392)}{15.661}$ & $\underset{(0.277)}{11.917}$ & $\underset{(1.449)}{15.611}$\\
4 & $\underset{(0.921)}{10.225}$ & $\underset{(3.270)}{79.029}$ & $\underset{(0.727)}{11.220}$ & $\underset{(5.567)}{77.689}$ & $\underset{(0.388)}{11.840}$ & $\underset{(3.536)}{79.076}$\\
5 & $\underset{(1.112)}{9.607}$ & $\underset{(3.794)}{79.557}$ & $\underset{(0.857)}{11.251}$ & $\underset{(8.175)}{78.925}$ & $\underset{(0.060)}{11.996}$ & $\underset{(1.888)}{85.454}$\\
6 & $\underset{(0.854)}{10.662}$ & $\underset{(1.498)}{15.809}$ & $\underset{(0.643)}{11.461}$ & $\underset{(1.401)}{15.820}$ & $\underset{(0.222)}{11.948}$ & $\underset{(1.421)}{15.889}$\\
\hline
\end{tabular}
\end{center}
\label{tbl:modsel}
\end{table}

We can see from table \ref{tbl:modsel} that the adaptive Lasso frequently selects the correct set of ``non-zero'' coefficients with small changes due to correlation, distinct errors specifications and number of candidate variables, these effects being more pronounced in small samples. The method performs well even in small to moderate samples. However, the sensibility of the model selection method for selecting the ``zero'' coefficients is affected by the number of candidate variables and error structure. We can see that the proportions of ``zero''-parameters correctly selected is smaller in the case we have many parameters and, particularly, when there is a AR(1) structure in the error term. Comparing models 4 and 5, we see that the combination of correlated errors and correlated variables has a large effect on the number of correctly selected ``zero''-coefficients in larger samples.

\subsubsection{Estimation Accuracy:}
We evaluate the estimation accuracy of the ``non-zero'' parameters and the standard deviation of the ``non-zero'' parameter estimates. For the estimation accuracy of the parameters, we compare the mean squared error (MSE) of the estimated parameters with the mean square error of the ``oracle-OLS'' parameters; and for the estimation accuracy of the parameter standard deviation we compare the estimate calculates by using \eqref{eq:covest} and the standard error calculated using resampling. We present the results for $(\beta_1,\beta_3,\beta_5,\gamma_1,\gamma_3,\gamma_5)$ for all six models.

Tables \ref{tbl:mse-mod1} -- \ref{tbl:mse-mod6} show the MSE of the parameters estimates. As expected the number of candidate variables, the covariance structure and the error structure affect the estimates. In small samples the standard error of the estimates are much larger than the oracle, however the mean square error quickly converges to the oracle MSE, as expected from theorem \ref{thm:oracle}. The worst performance was model 4 that showed an MSE  of the $\beta$ estimates almost three time as big as the oracle in moderate-to-large samples (200 observations), however the decrease in the MSE is very steep, indicating that this difference vanishes in larger samples. In fact, this error is really small in larger samples, being negligible when we have 1000 observations.

\begin{table}
	\caption{MSE: Model 1}
	\begin{center}
	\begin{tabular}{lcccccc}
	\hline
	&\multicolumn{2}{c}{50}&\multicolumn{2}{c}{100}&\multicolumn{2}{c}{200}\\
	Parameter &AdaLasso &Oracle-OLS &AdaLasso &Oracle-OLS &AdaLasso &Oracle-OLS\\
	\hline
	$\beta_1$ & $0.117$ & $0.045$ & $0.018$ & $0.010$ & $0.004$ & $0.002$\\
	$\beta_3$ & $0.129$ & $0.062$ & $0.022$ & $0.012$ & $0.004$ & $0.003$\\
	$\beta_5$ & $0.131$ & $0.051$ & $0.022$ & $0.013$ & $0.003$ & $0.003$\\
	$\gamma_1$ & $0.148$ & $0.087$ & $0.052$ & $0.033$ & $0.024$ & $0.017$\\
	$\gamma_3$ & $0.158$ & $0.102$ & $0.055$ & $0.045$ & $0.023$ & $0.019$\\
	$\gamma_5$ & $0.154$ & $0.113$ & $0.065$ & $0.042$ & $0.025$ & $0.021$\\
	\hline
	\end{tabular}
	\end{center}
\label{tbl:mse-mod1}
\end{table}

\begin{table}
	\caption{MSE: Model 2}
	\begin{center}
	\begin{tabular}{lcccccc}
	\hline
	&\multicolumn{2}{c}{50}&\multicolumn{2}{c}{100}&\multicolumn{2}{c}{200}\\
	Parameter &AdaLasso &Oracle-OLS &AdaLasso &Oracle-OLS &AdaLasso &Oracle-OLS\\
	\hline
$\beta_1$ & $0.889$ & $0.171$ & $0.111$ & $0.035$ & $0.022$ & $0.009$\\
$\beta_3$ & $0.834$ & $0.323$ & $0.138$ & $0.065$ & $0.025$ & $0.017$\\
$\beta_5$ & $0.329$ & $0.309$ & $0.138$ & $0.074$ & $0.028$ & $0.015$\\
$\gamma_1$ & $1.152$ & $0.392$ & $0.365$ & $0.146$ & $0.122$ & $0.061$\\
$\gamma_3$ & $1.002$ & $0.549$ & $0.454$ & $0.257$ & $0.167$ & $0.125$\\
$\gamma_5$ & $0.373$ & $0.660$ & $0.232$ & $0.249$ & $0.154$ & $0.114$\\
\hline
	\end{tabular}
	\end{center}
	\label{tbl:mse-mod2}
\end{table}

\begin{table}
	\caption{MSE: Model 3}
	\begin{center}
	\begin{tabular}{lcccccc}
	\hline
	&\multicolumn{2}{c}{50}&\multicolumn{2}{c}{100}&\multicolumn{2}{c}{200}\\
	Parameter &AdaLasso &Oracle-OLS &AdaLasso &Oracle-OLS &AdaLasso &Oracle-OLS\\
	\hline
$\beta_1$ & $0.192$ & $0.105$ & $0.058$ & $0.036$ & $0.020$ & $0.011$\\
$\beta_3$ & $0.196$ & $0.129$ & $0.067$ & $0.044$ & $0.021$ & $0.014$\\
$\beta_5$ & $0.174$ & $0.128$ & $0.077$ & $0.050$ & $0.020$ & $0.013$\\
$\gamma_1$ & $0.142$ & $0.100$ & $0.059$ & $0.044$ & $0.028$ & $0.023$\\
$\gamma_3$ & $0.155$ & $0.117$ & $0.061$ & $0.054$ & $0.029$ & $0.028$\\
$\gamma_5$ & $0.138$ & $0.115$ & $0.079$ & $0.056$ & $0.035$ & $0.029$\\
\hline
	\end{tabular}
	\end{center}
	\label{tbl:mse-mod3}
\end{table}

\begin{table}
	\caption{MSE: Model 4}
	\begin{center}
	\begin{tabular}{lcccccc}
	\hline
	&\multicolumn{2}{c}{50}&\multicolumn{2}{c}{100}&\multicolumn{2}{c}{200}\\
	Parameter &AdaLasso &Oracle-OLS &AdaLasso &Oracle-OLS &AdaLasso &Oracle-OLS\\
	\hline
$\beta_1$ & $0.334$ & $0.098$ & $0.087$ & $0.033$ & $0.029$ & $0.012$\\
$\beta_3$ & $0.282$ & $0.116$ & $0.100$ & $0.046$ & $0.028$ & $0.013$\\
$\beta_5$ & $0.190$ & $0.117$ & $0.096$ & $0.045$ & $0.032$ & $0.012$\\
$\gamma_1$ & $0.247$ & $0.104$ & $0.070$ & $0.039$ & $0.026$ & $0.022$\\
$\gamma_3$ & $0.216$ & $0.131$ & $0.080$ & $0.059$ & $0.029$ & $0.028$\\
$\gamma_5$ & $0.175$ & $0.116$ & $0.105$ & $0.053$ & $0.029$ & $0.030$\\
\hline
	\end{tabular}
	\end{center}
	\label{tbl:mse-mod4}
\end{table}

\begin{table}
	\caption{MSE: Model 5}
	\begin{center}
	\begin{tabular}{lcccccc}
	\hline
	&\multicolumn{2}{c}{50}&\multicolumn{2}{c}{100}&\multicolumn{2}{c}{200}\\
	Parameter &AdaLasso &Oracle-OLS &AdaLasso &Oracle-OLS &AdaLasso &Oracle-OLS\\
	\hline
$\beta_1$ & $0.404$ & $0.043$ & $0.070$ & $0.010$ & $0.005$ & $0.002$\\
$\beta_3$ & $0.341$ & $0.049$ & $0.082$ & $0.012$ & $0.004$ & $0.003$\\
$\beta_5$ & $0.208$ & $0.057$ & $0.103$ & $0.013$ & $0.004$ & $0.003$\\
$\gamma_1$ & $0.392$ & $0.060$ & $0.055$ & $0.029$ & $0.013$ & $0.012$\\
$\gamma_3$ & $0.385$ & $0.064$ & $0.059$ & $0.031$ & $0.012$ & $0.012$\\
$\gamma_5$ & $0.191$ & $0.061$ & $0.075$ & $0.026$ & $0.015$ & $0.012$\\
\hline
	\end{tabular}
	\end{center}
	\label{tbl:mse-mod5}
\end{table}

\begin{table}
	\caption{MSE: Model 6}
	\begin{center}
	\begin{tabular}{lcccccc}
	\hline
	&\multicolumn{2}{c}{50}&\multicolumn{2}{c}{100}&\multicolumn{2}{c}{200}\\
	Parameter &AdaLasso &Oracle-OLS &AdaLasso &Oracle-OLS &AdaLasso &Oracle-OLS\\
	\hline
	$\beta_1$ & $0.173$ & $0.099$ & $0.054$ & $0.032$ & $0.021$ & $0.010$\\
	$\beta_3$ & $0.169$ & $0.114$ & $0.057$ & $0.041$ & $0.018$ & $0.012$\\
	$\beta_5$ & $0.165$ & $0.117$ & $0.066$ & $0.040$ & $0.015$ & $0.012$\\
	$\gamma_1$ & $0.133$ & $0.089$ & $0.046$ & $0.040$ & $0.021$ & $0.019$\\
	$\gamma_3$ & $0.147$ & $0.119$ & $0.056$ & $0.053$ & $0.024$ & $0.027$\\
	$\gamma_5$ & $0.140$ & $0.115$ & $0.072$ & $0.049$ & $0.031$ & $0.021$\\
	\hline
	\end{tabular}
	\end{center}
	\label{tbl:mse-mod6}
\end{table}

Tables \ref{tbl:sd-mod1} -- \ref{tbl:sd-mod6} compare the estimated standard deviation (SD) of the parameter with the actual standard deviation of the parameter calculated using resampling. We estimate $\sigma_{uu}$ and $\sigma_{uu}^*$ assuming knowledge of the data generating process of the error term, which is a reasonable assumption since we are only interested in verifying the behavior of the proposed formula in finite samples. If the data generating process is unknown, we can estimate the autoregressive order using the same method proposed here.

We can see that, for all model specifications, the difference between the estimated standard deviations calculated using resampling and equation \eqref{eq:covest} shrink as the sample size increases for both $\beta$ and $\gamma$. The worst performance was model 2, where the variables are highly correlated. In larger samples the estimated standard deviation is reasonably close to the ``true'' one estimated by using resampling.

\begin{table}
	\caption{Model 1: Standard Deviation and Estimated Standard Deviation}
	\label{tbl:sd-mod1}
	\begin{center}
\begin{tabular}{lcccccc}
\hline
&\multicolumn{2}{c}{50}&\multicolumn{2}{c}{100}&\multicolumn{2}{c}{200}\\
Parameter&$\sigma$ &$\hat{\sigma}$ &$\sigma$ &$\hat{\sigma}$ &$\sigma$ &$\hat{\sigma}$\\
\hline
$\beta_1$ & $0.287$ & $0.165$ & $0.128$ & $0.080$ & $0.053$ & $0.040$\\
$\beta_3$ & $0.333$ & $0.181$ & $0.132$ & $0.090$ & $0.063$ & $0.046$\\
$\beta_5$ & $0.374$ & $0.109$ & $0.157$ & $0.083$ & $0.064$ & $0.045$\\
$\gamma_1$ & $0.356$ & $0.276$ & $0.194$ & $0.172$ & $0.127$ & $0.121$\\
$\gamma_3$ & $0.406$ & $0.296$ & $0.222$ & $0.190$ & $0.147$ & $0.135$\\
$\gamma_5$ & $0.404$ & $0.169$ & $0.273$ & $0.143$ & $0.168$ & $0.121$\\
\hline
\end{tabular}

	\end{center}
\end{table}

\begin{table}
\caption{Model 2: Standard Deviation and Estimated Standard Deviation}
\label{tbl:sd-mod2}
\begin{center}

\begin{tabular}{lcccccc}
\hline
&\multicolumn{2}{c}{50}&\multicolumn{2}{c}{100}&\multicolumn{2}{c}{200}\\
Parameter&$\sigma$ &$\hat{\sigma}$ &$\sigma$ &$\hat{\sigma}$ &$\sigma$ &$\hat{\sigma}$\\
\hline
$\beta_1$ & $0.576$ & $0.345$ & $0.270$ & $0.179$ & $0.111$ & $0.084$\\
$\beta_3$ & $0.919$ & $0.358$ & $0.372$ & $0.222$ & $0.152$ & $0.110$\\
$\beta_5$ & $0.637$ & $0.163$ & $0.434$ & $0.111$ & $0.207$ & $0.089$\\
$\gamma_1$ & $0.682$ & $0.623$ & $0.404$ & $0.396$ & $0.253$ & $0.254$\\
$\gamma_3$ & $1.048$ & $0.563$ & $0.650$ & $0.450$ & $0.368$ & $0.324$\\
$\gamma_5$ & $0.739$ & $0.210$ & $0.586$ & $0.168$ & $0.451$ & $0.130$\\
\hline
\end{tabular}
	\end{center}
\end{table}

\begin{table}
\caption{Model 3: Standard Deviation and Estimated Standard Deviation}
\label{tbl:sd-mod3}
\begin{center}

\begin{tabular}{lcccccc}
\hline
&\multicolumn{2}{c}{50}&\multicolumn{2}{c}{100}&\multicolumn{2}{c}{200}\\
Parameter&$\sigma$ &$\hat{\sigma}$ &$\sigma$ &$\hat{\sigma}$ &$\sigma$ &$\hat{\sigma}$\\
\hline
$\beta_1$ & $0.399$ & $0.389$ & $0.226$ & $0.206$ & $0.124$ & $0.107$\\
$\beta_3$ & $0.450$ & $0.417$ & $0.251$ & $0.228$ & $0.136$ & $0.114$\\
$\beta_5$ & $0.436$ & $0.253$ & $0.281$ & $0.189$ & $0.145$ & $0.115$\\
$\gamma_1$ & $0.380$ & $0.341$ & $0.231$ & $0.240$ & $0.153$ & $0.166$\\
$\gamma_3$ & $0.370$ & $0.371$ & $0.232$ & $0.263$ & $0.172$ & $0.184$\\
$\gamma_5$ & $0.406$ & $0.213$ & $0.291$ & $0.191$ & $0.186$ & $0.166$\\
\hline
\end{tabular}
	\end{center}
\end{table}

\begin{table}
\caption{Model 4: Standard Deviation and Estimated Standard Deviation}
\label{tbl:sd-mod4}
\begin{center}

\begin{tabular}{lcccccc}
\hline
&\multicolumn{2}{c}{50}&\multicolumn{2}{c}{100}&\multicolumn{2}{c}{200}\\
Parameter&$\sigma$ &$\hat{\sigma}$ &$\sigma$ &$\hat{\sigma}$ &$\sigma$ &$\hat{\sigma}$\\
\hline
$\beta_1$ & $0.479$ & $0.615$ & $0.301$ & $0.367$ & $0.155$ & $0.152$\\
$\beta_3$ & $0.512$ & $0.673$ & $0.332$ & $0.405$ & $0.171$ & $0.170$\\
$\beta_5$ & $0.470$ & $0.494$ & $0.329$ & $0.338$ & $0.185$ & $0.166$\\
$\gamma_1$ & $0.380$ & $0.555$ & $0.270$ & $0.428$ & $0.148$ & $0.246$\\
$\gamma_3$ & $0.429$ & $0.606$ & $0.284$ & $0.477$ & $0.166$ & $0.274$\\
$\gamma_5$ & $0.427$ & $0.328$ & $0.315$ & $0.348$ & $0.181$ & $0.245$\\
\hline
\end{tabular}
	\end{center}
\end{table}

\begin{table}
\caption{Model 5: Standard Deviation and Estimated Standard Deviation}
\label{tbl:sd-mod5}
\begin{center}

\begin{tabular}{lcccccc}
\hline
&\multicolumn{2}{c}{50}&\multicolumn{2}{c}{100}&\multicolumn{2}{c}{200}\\
Parameter&$\sigma$ &$\hat{\sigma}$ &$\sigma$ &$\hat{\sigma}$ &$\sigma$ &$\hat{\sigma}$\\
\hline
$\beta_1$ & $0.543$ & $0.301$ & $0.226$ & $0.125$ & $0.064$ & $0.040$\\
$\beta_3$ & $0.562$ & $0.329$ & $0.261$ & $0.140$ & $0.070$ & $0.046$\\
$\beta_5$ & $0.491$ & $0.218$ & $0.307$ & $0.113$ & $0.082$ & $0.046$\\
$\gamma_1$ & $0.452$ & $0.445$ & $0.233$ & $0.239$ & $0.113$ & $0.106$\\
$\gamma_3$ & $0.511$ & $0.395$ & $0.236$ & $0.238$ & $0.121$ & $0.106$\\
$\gamma_5$ & $0.359$ & $0.163$ & $0.260$ & $0.190$ & $0.131$ & $0.100$\\
\hline
\end{tabular}
	\end{center}
\end{table}

\begin{table}
\caption{Model 6: Standard Deviation and Estimated Standard Deviation}
\label{tbl:sd-mod6}
\begin{center}

\begin{tabular}{lcccccc}
\hline
&\multicolumn{2}{c}{50}&\multicolumn{2}{c}{100}&\multicolumn{2}{c}{200}\\
Parameter&$\sigma$ &$\hat{\sigma}$ &$\sigma$ &$\hat{\sigma}$ &$\sigma$ &$\hat{\sigma}$\\
\hline
$\beta_1$ & $0.356$ & $0.339$ & $0.217$ & $0.196$ & $0.117$ & $0.100$\\
$\beta_3$ & $0.425$ & $0.386$ & $0.243$ & $0.215$ & $0.128$ & $0.110$\\
$\beta_5$ & $0.420$ & $0.248$ & $0.277$ & $0.175$ & $0.140$ & $0.111$\\
$\gamma_1$ & $0.334$ & $0.307$ & $0.206$ & $0.225$ & $0.148$ & $0.159$\\
$\gamma_3$ & $0.370$ & $0.332$ & $0.230$ & $0.247$ & $0.163$ & $0.176$\\
$\gamma_5$ & $0.406$ & $0.185$ & $0.280$ & $0.181$ & $0.170$ & $0.160$\\
\hline
\end{tabular}
	\end{center}
\end{table}

\subsubsection{Prediction Accuracy:}
We evaluate the prediction accuracy by calculating prediction mean square error\footnote{$PMSE = K^{-1}\sum_{t=T+1}^{T+K}(y_t-\hat{y}_t)^2$, where $\hat{y}_t$ is the predicted value of $y_t$ using the estimated parameters.} (PMSE) for each model and dividing by the ``oracle-OLS'' PMSE, i.e. the PMSE of the OLS estimator conditional on knowing the variables that enter in the model. This measure tells us how close we are from the traditional OLS predictor, a number close to 1 means that the prediction accuracy is very close to the oracle prediction. To avoid the effect of large values, we used the average median of the PMSEs, estimated using resampling. Table \ref{tbl:pmse} summarizes the results.

\begin{table}[h]
	\caption{Predicton Mean Squared Error}
	\begin{center}
		\begin{tabular}{lccc}
			\hline
			Model &50 &100 &200\\
			\hline
			1 & $1.640$ & $1.101$ & $1.022$\\
			2 & $1.516$ & $1.174$ & $1.075$\\
			3 & $1.559$ & $1.418$ & $1.329$\\
			4 & $4.524$ & $4.887$ & $4.362$\\
			5 & $7.297$ & $4.188$ & $1.120$\\
			6 & $1.442$ & $1.721$ & $1.343$\\
			\hline
		\end{tabular}
	\end{center}
	\label{tbl:pmse}
\end{table}

We can see that the PMSE approaches the oracle PMSE as the sample size increases. The rate in which the prediction error decreases depends on the number of candidate variables and the error structure, for instance, in models 4 and 5 the PMSE can be as much as 7 times larger than the oracle in small samples, but this error rapidly converges to the oracle in the case where the errors are i.i.d. and the candidate variables uncorrelated with the variables in the model. 

In model 4 the relative PMSE is very large and decreases slowly. This behavior can be explained by observing the performance of the method in choosing the ``zero'' parameters in this model. We can see that although the model selects the correct set of ``non-zero'' parameters correctly, a number of ``zero'' parameters is also selected and, since we are dealing with ``explosive'' regressors, the model prediction variance also increases. However, as the sample size increases the relative error also decreases as expected, for instance for sample sizes $500$ and $1000$, the relative PMSE are respectively $3.837$ and $3.013$.

\section{Conclusion}\label{conclusion}

In this paper, we provide an extension of the Adaptive Lasso variable selection method to cointegrated regressions. We show that, under some regularity conditions frequently assumed in the model selection literature and cointegration literature, the method selects the correct subset of variables and converges to the ``oracle'' estimate, i.e. the estimator under the assumption we know the variables that enter in the model.

Although the result only allows for a sub-linear number of $I(1)$ candidate variables and a polynomial number of canditate $I(0)$ variables. We allow the number of $I(0)$ variables that enter in the model to increase with the sample size $T$. Such condition allow for Dynamic OLS Estimation if we consider the integrated variables to be endogenous.  Another interesting extension is the multivariate case. We can see that all results hold for the vector case if the dimension of $y_t$ is fixed, i.e., a fixed number of regressions. It can be shown by just adapting the proof of the theorems and conditions to the vector case. 

All the previous result hold if all parameters $\beta = 0$ or $\gamma = 0$, meaning that we do not need $I(1)$ or $I(0)$ variables for the results to hold. Also, the inclusion of the intercept does not change our results.

\bibliographystyle{abbrvnat}
\bibliography{cointegration,variableselection}

\appendix
\section{Proof of theorems \ref{thm:modsel} and \ref{thm:oracle}}
Before presenting the proof of Theorems \ref{thm:modsel} and \ref{thm:oracle} we introduce an useful lemma.

\begin{lemma}
	Let 
		\begin{equation}
			\Omega_\infty = \left(
			\begin{array}{cc}
				\Omega_{X,\infty}&\bs{0}\\
				\bs{0}'&\Omega_{Z,\infty}
			\end{array}
			\right)
			\label{eq:omega8}
		\end{equation}
		where
		\[
			\Omega_{X,\infty} = \int_0^1B_{X(1)}(r)B_{X(1)}'(r)dr\quad\mbox{ and }\quad\Omega_{Z,\infty} = \Sigma_{Z(1)^2},
		\]
		where for any $0\le r\le1$, $B_{X(1)}(r) = \lim_{T\rightarrow\infty}T^{-1/2}\sum_{t=1}^{\lceil rT \rceil}\bs{v}_t(1)$.
		Similarly, split the matrix $\Omega_{11}$ into
		\begin{equation}
			\Omega_{11} = \left(
			\begin{array}{cc}
				\Omega_{X(1)^2}&\Omega_{Z(1)X(1)}\\
				\Omega_{Z(1)X(1)}'&\Omega_{Z(1)^2}
			\end{array}
			\right) = \left(
			\begin{array}{cc}
				T^{-2}X(1)'X(1)   & T^{-3/2}Z(1)'X(1)\\
				T^{-3/2}X(1)'Z(1) & T^{-1}Z(1)'Z(1)
			\end{array}
			\right).
			\label{eq:omega11}
		\end{equation}
		Let $\delta=(\delta_1',\delta_2')'$ and $\xi = (\xi_1',\xi_2')'$ denote a couple of $(q_1+q_2)\times 1$ vectors satisfying $\delta_i'\delta_i\le q_i$ and $\xi_i'\xi_i\le q_i$ for $i=1,2$. Then under Assumption \ref{a:ip} and if $q_1=O(1)$ and $q_2=o(T^{1/2})$, we have
	\begin{itemize}
		\item [(a)] $\delta'(\Omega_{11}-\Omega_\infty)\xi = o_p(1) $;
		\item [(b)] $\delta_1'(\Omega_{X(1)^2}-\Omega_{X,\infty})\xi_1 = o_p(1)$;
		\item [(c)] $\delta_2'(\Omega_{Z(1)^2}-\Omega_{Z,\infty})\xi_2 = o_p(1)$; and
		\item [(d)] $\delta_2'\Omega_{Z(1)X(1)}\xi_1 =o_p(1)$.
	\end{itemize}
		
	\label{l:omega}
\end{lemma}

\begin{proof}
Let's first consider the off-diagonal elements $\Omega_{X(1)Z(1)}=T^{-3/2}X(1)'Z(1)$. We have 
\begin{align*}
	\sup_{\|\delta_1\|^2\le q_1 \, ,\|\xi_2\|^2\le q_2}\delta_1'(T^{-3/2}X(1)'Z(1))\xi_2 &= T^{-1/2}\sup_{\|\delta_1\|^2\le q_1 \, ,\|\xi_2\|^2\le q_2}\sum_{i=1}^{q_1}\sum_{j=1}^{q_2}\delta_{1i}\xi_{2j}(T^{-1}\sum_{t=1}^{T}x_{it}z_{it})\\
	&\le \frac{q_1q_2}{T^{1/2}}O_p(1)\\
	&= \frac{q_2}{T^{1/2}}O(1)O_p(1)\\
	&= o_p(1)
\end{align*}
because $q_2/T^{1/2}=o(1)$.

It from classical results in cointegration theory that the element $|\Omega_{X(1)^2}-\Omega_{X,\infty}|=o_p(1)$ since $q_1=O(1)$. Finally, we have to show that $\delta_1'(\Omega_{Z(1)^2}-\Omega_{Z,\infty})\xi_1 = o_p(1)$. Note that $G_T = \sqrt{T}\delta_2'(\Omega_{Z(1)^2}-\Omega_{Z,\infty})\xi_2$ is a centered empirical process and that for any $\varepsilon>0$,
\begin{align*}
	\Pr\left( \delta_2'(\Omega_{Z(1)^2}-\Omega_{z,\infty})\xi_2 > \varepsilon \right) &= \Pr\left( G_T \ge \sqrt{T}\varepsilon \right)\\
	&\le\E(G_T)^2/T\varepsilon^2\\
	&\le \frac{q_2^2\max_{1\le i\le j\le q_2} \E(T^{-1/2}\sum z_{it}z_{jt}-\sigma_{ij})^2}{\varepsilon T}\\
	&=\frac{q_2^2}{\varepsilon T} O(1)\\
	&\rightarrow 0.
\end{align*}

Finally, combining these three results we have $\delta'(\Omega_{11}-\Omega_\infty)\xi = o_p(1)$, proving the lemma.
\end{proof}

\begin{proof}[Proof of theorem \ref{thm:modsel}]
We knoe from proposition \ref{p:necessary} that showing sign consistency is equivalent to showing that $\Pr(\mathcal{A}_T\cap\mathcal{B}_T)\rightarrow 1$.  It is sufficient to show that $1-\Pr(\mathcal{A}^c_T) - \Pr(\mathcal{B}^c_T) \rightarrow 1$, the superscript ``$c$'' meaning complement. 

The proof is divided in two parts. In the first one we show that $\Pr(\mathcal{A}_T^c)\rightarrow 0$ and in the second part we show that $\Pr(\mathcal{B}_T^c)\rightarrow 0$.

Note the event $\mathcal{A}_T^c$ is given by
\[
\mathcal{A}_T = \left\{ \Gamma^{-1/2}|\Omega_{11}^{-1}W(1)'U| < \Gamma^{1/2}|\theta_0(1)|-\frac{1}{2}\Gamma^{-1/2}\lambda|\Omega_{11}^{-1}L(1)\sgn(\theta_0(1))| \right\}\]
where the inequality holds elementwise. Hence, the complement is an union and can be split into $\mathcal{A}_T^c(X)\cup\mathcal{A}_T^c(Z)$, with the events $\mathcal{A}_T^c(X)$ and $\mathcal{A}_T^c(Z)$ given by
\begin{multline*}
\mathcal{A}_T^c(X) = \left\{ T^{-1}|[\Omega_{X(1)^2}^{-1}+o_p(1)]X(1)'U| > T|\beta_0(1)|\right.\\
\left.-\frac{1}{2}T^{-1}\lambda_1|[\Omega_{X(1)^2}^{-1}+o_p(1)]L_X(1)\sgn(\beta_0(1))| \right\}
\end{multline*}
and 
\begin{multline*}
\mathcal{A}_T^c(Z) = \left\{ T^{-1/2}|[\Omega_{Z(1)^2}^{-1}+o_p(1)]Z(1)'U| > T^{1/2}|\gamma_0(1)|\right.\\
\left.-\frac{1}{2}T^{-1/2}\lambda_2|[\Omega_{Z(1)^2}^{-1}+o_p(1)]L_Z(1)\sgn(\gamma_0(1))| \right\}
\end{multline*}

We first deal with $\mathcal{A}_T^c(X)$. By Assumptions \ref{a:adalasso}, \ref{a:lasso} and \ref{a:ip}, and by using Lemma \ref{l:omega} we jave
\begin{align*}
	T^{-1}\lambda_1\lambda_{1j}|[(\Omega_{X(1)^2}^{-1}+o_p(1))\sgn(\beta_{0}(1))| ]_j&=\frac{\lambda_1}{T^{1+\rho}}V_{1j}[|(\Omega^{-1}_{X,\infty})\sgn(\beta_0(1))|]_j|+o_p(1),\\
& = o_p(1),
\end{align*}
where the first line follows from lemma \ref{l:omega} and Assumption \ref{a:adalasso} ($T^{\rho}\lambda_{1j}=V_{1j}+o_p(1)$) and the last line follows from $\lambda_1=o(T^{1+\rho})$ and the fact that $|(\Omega^{-1}_{X,\infty})\sgn(\beta_0(1))|_j = q_1[O_p(1)+o_p(1)] = O_p(1)$.

Hence,
\begin{align*}
	\Pr\left(\mathcal{A}_T^c(X)\right) &= \Pr\left(\left\{ \left[|T^{-1}\Omega^{-1}_{X,\infty}X(1)'U\right]_j > T|\beta_{0j}| \right\},\,j=1,\dots,q_1\right)+o_p(1)\\
	&\le \sum_{j=1}^{q_1}\Pr\left(\left[|T^{-1}\Omega^{-1}_{X,\infty}X(1)'U\right]_j > T|\beta_{0j}| \right)+o_p(1)\\
	&\le \frac{q_1}{T^2\beta_*^2}\max_{1\le j\le q_1}\E\left( \left[ T^{-1}|\Omega^{-1}_{X,\infty}X(1)'U|\right]^2_j \right) +o_p(1)\\
	&\rightarrow 0,
\end{align*}
where the second line follows from the union bound, third line from the Chebyschev's inequality and the last line by Assumption \ref{a:ip} and because $q_1$ is constant.

Now we focus our attention on $\Pr(\mathcal{A}_T^c(Z))$. First denote by $\mathcal{D}_T$ the event $\{\|\delta\|^2=q_2:\,\delta'| (T^{-1}Z(1)Z(1))^{-1} - \Omega_{Z,\infty}^{-1})|\delta>\varepsilon\tau_*^{-1}\}$, for $\varepsilon+1<c_\varepsilon|\gamma^*|$ and $c_\varepsilon$ some positive constant. We have alreadu shown that $\Pr(\mathcal{D}_T)\rightarrow 0$ as $T\rightarrow \infty$. Consider the spectral decomposition of $\Omega_{Z,\infty} = EDE'$ with $E$ a matrix of $q_2$ eigenvectors and $D$ a diagonal matrix of eigenvalues. By assumption the elements of $D$ are greater than $\tau_*$, then inside $\mathcal{D}_T^c$ and for all $j=1,\dots,q_2$,
\begin{align*}
	T^{-1/2}\lambda_2\lambda_{2j}[|(\Omega^{-1}_{Z,\infty} + \varepsilon/\tau_*) \sgn(\gamma_0(1)|]_j &=  T^{-1/2}\lambda_2\lambda_{2j}[|ED^{-1}E'\sgn(\gamma_0(1))|]_j+ T^{-1/2}\lambda_2\lambda_{2j}q_2\varepsilon/\tau_*\\
	&\le T^{-1/2}q_2\lambda_2\lambda_{2j}/\tau_*+ T^{-1/2}\lambda_2\lambda_{2j}q_2\varepsilon/\tau_*\\
	&= \left(1+\varepsilon\right)\frac{\lambda_2 q_2}{\tau_*T^{(1+\rho)/2}}V_{2j}(1+o_p(1))\\
	& \le c_\varepsilon\gamma^*\frac{\lambda_2 q_2}{\tau_*T^{(1+\rho)/2}}V_{2j}(1+o_p(1)),
\end{align*}
where the second line follows from
\begin{align*}
	[|\Omega_{z,\infty}^{-1}\sgn(\gamma_0(1))|]^2_j &\le \sup_{\|\delta\|=1}\left( |\delta'[\Omega_{z,\infty}^{-1}\sgn(\gamma_0(1)]| \right)^2\\
	&\le \sup_{\|\delta\|=1}\|\delta\|^2\|\Omega_{z,\infty}^{-1}\sgn(\gamma_0(1))\|^2\\
	&=\sgn(\gamma_0(1))'ED^{-2}E'\sgn(\gamma_0(1))\\
	&\le \|\sgn(\gamma_0(1))\|^2\|E\|^2\tau_*^{-2}\\
	&\le q_1^2\tau_*^{-2}
\end{align*}
and the third line from the assumption that $T^{\rho/2}\lambda_{2j}$ converges to a stationary process.

Then,  
\begin{align*}
	\Pr \left( \mathcal{A}^c_T(Z) \cap \mathcal{D}_T^c\right) &\le \Pr\left( \max_{1\le j\le q_2}[|T^{-1/2}\Omega_{Z,\infty}^{-1}Z(1)'U|]_j>T^{1/2}|\gamma_*|-c_\varepsilon\gamma^*q_2\lambda_2T^{-(1/2+\rho)}V_2\tau_*^{-1} \right)\\
	&\le \frac{\gamma_*^{2}}{T}\frac{\E\left[\max_{1\le j\le q_2}[|T^{-1/2}\Omega_{Z,\infty}^{-1}Z(1)'U|]_j^2\right]}{(1-c_\varepsilon\lambda_2q_2V_2/\tau_*T^{1+\rho/2})^2}\\
	&\le \frac{\gamma_*^2}{T}\frac{q_2^{2+1/d}\tau_*^{-2}\max_j\left(\E|\sum_{t=1}^Tz_{jt}u_t|^{2d}\right)^{1/d}  }{\left[ 1-\left( \frac{c_\varepsilon}{\tau_*}\frac{\lambda_2}{T^{(1+\rho)/2}}\frac{q_2}{T^{1/2}}V_2 \right) \right]^2} \\
	&\rightarrow 0,
\end{align*}
where the second line from the Chebyschev's inequality. The third line follows from the bound
\begin{align*}
	\left(\max_j[T^{-1/2}\Omega^{-1}_{Z,\infty}Z(1)'U]_j\right)^2 &= \max_j\left( [T^{-1/2}ED^{-1}E'Z(1)'U]_j \right)^2\\
	&\le \tau_*^{-2}q_2^2(\max_j[T^{-1/2}Z(1)'U]_j)^2,\\
\end{align*}
 and by the Jensen's inequality, $\E(\max_j|T^{-1/2}\sum_{t=1}^{T}z_{jt}u_t|^2) \le q_2^{1/d}\max_j\left( \E|T^{-1/2}\sum_{t=1}^Tz_{jt}u_t|^{2d} \right)^{1/d}$. 
The conclusion follows from assumptions \ref{a:ip}, \ref{a:lasso} and \ref{a:adalasso}.

Moving to $\mathcal{B}_T^c$, it follows from Lemma \ref{l:omega} that $M(1) = M_\infty(1) + o_p(1)$, and the matrix $M_\infty(1) =\diag(M_{X}(1),M_{Z}(1))$, with
\[
M_{X}(1) = \bs{I}_T-X(1)(X(1)'X(1))^{-1}X(1)' \quad \mbox{ and }\quad M_{Z}(1) = \bs{I}_T- Z(1)(Z(1)'Z(1))^{-1}Z(1)'.
\]

The events $\mathcal{B}_T^c(X)$ and $\mathcal{B}_T^c(Z)$ can be written as
\begin{multline*}
\mathcal{B}_T^c(X) =  \left\{ \max_{q_1<j\le n_1}|2T^{-1}\bs{x}_j'[M_X(1)+o_p(1)]U| \right.\\
\left.> T^{-1}\lambda_1\lambda_{1j}-\lambda_1|T^{-1}\bs{x}_j'X(1)[\Omega_{X,\infty}^{-1}+o_p(1)]L_X(1)\sgn(\beta_0(1))| \right\},
\end{multline*}
and
\begin{multline*}
\mathcal{B}_T^c(Z) =  \left\{ \max_{q_2<j\le n_2}|2T^{-1/2}\bs{z}_j'[M_{Z}(1)+o_p(1)]U| \right.\\
\left.> T^{-1/2}\lambda_1\lambda_{2j}-\lambda_2|T^{-1/2}\bs{z_j}'Z(1)[\Omega_{Z,\infty}^{-1}+o_p(1)]L_Z(1)\sgn(\gamma_0(1))| \right\}.
\end{multline*}

We further consider the event $\mathcal{C}_T(X)=\{\max_{1\leq j\leq q_1} \lambda_{1j} < \beta_*^{-1}\}$ and $\mathcal{C}_T(Z)=\{\max_{1\leq j\leq q_2} \lambda_{2j} < \gamma_*^{-1}\}$, then
\begin{subequations}
	\begin{align}
	\Pr(\mathcal{B}_T^c(X)) &\le \Pr(\mathcal{B}_T^c(X) \cap \mathcal{C}_T) + \Pr(\mathcal{C}_T^c(X)),\label{eq:BcX}\\
	\Pr(\mathcal{B}_T^c(Z)) &\le \Pr(\mathcal{B}_T^c(Z) \cap \mathcal{C}_T) + \Pr(\mathcal{C}_T^c(Z)).\label{eq:BcZ}
	\end{align}
\end{subequations}

By the Weak Irrepresentable Condition, one has inside $\mathcal{C}_T(X)$
\[
T^{-1}\lambda_1|\bs{x}_j'X(1)[\Omega_{X,\infty}^{-1}+o_p(1)]L_X(1)\sgn(\beta_0(1))| \le \frac{\lambda_1(\beta_*\lambda_{1j}-\eta)}{T\beta_*}+o_p(1),
\]
and hence, 
\[
T^{-1}\lambda_1\lambda_{1j}-T^{-1}\lambda_1|\bs{x}_j'X(1)[\Omega_{X,\infty}^{-1}+o_p(1)]L_X(1)\sgn(\beta_0(1))| \le \frac{\lambda_1\eta}{T\beta_*} +o_p(1).
\]
Therefore,
\begin{align*}
	\Pr\left( \mathcal{B}^c_T(X)\cap\mathcal{C}_T(X) \right)&\le \Pr\left( \max_{q_1+1\le j\le n_1}|2T^{-1}\bs{x}_jM_X(1)U| > \lambda_1\eta/T\beta_* \right)+o_p(1)\\
	&\le \frac{4\beta_*^2}{\eta^2}\E[\max_j|T^{-1}\bs{x}_jM_X(1)U|^2]\frac{T^2}{\lambda_1^2}+o_p(1)\\
	&\le \frac{4\beta_*^2\max_j\E|T^{-1}\bs{x}_j'U|^2}{\eta^2}\frac{m_1T^2}{\lambda_1^2}+o_p(1)\\
	&\rightarrow 0,
\end{align*}
where the second line follows by the Chebyschev's inequality, the third line from the fact that for any projection matrix $M$, 
\[
\E|\bs{x}_j'MU|^2= \E|\bs{x}_j'U|^2 - \E|\bs{x}_j(I-M)'U|^2 \le \E|\bs{x}_j'U|^2;
\]
and the last line from assumption \ref{a:lasso} and $q_1=O(1)$. 

Applying the same reasoning to $\mathcal{B}_T^C(Z)\cap\mathcal{C}_T(Z)$, the WIC gives us
\begin{align*}
	\Pr\left( \mathcal{B}^c_T(Z)\cap\mathcal{C}_T(Z) \right)&\le \Pr\left( \max_{q_2+1\le j\le n_2}|2T^{-1/2}\bs{z}_jM_Z(1)U| > \lambda_2\eta/T^{1/2}\gamma_* \right)+o_p(1)\\
	&\le \frac{4\gamma_*^2}{\eta^2}\E[\max_j|T^{-1/2}\bs{z}_jM_Z(1)U|^2]\frac{T}{\lambda_2^2}+o_p(1)\\
	&\le \frac{4\gamma_*^2c_d}{\eta^2}\frac{m_2^{1/d}T}{\lambda_2^2}+o_p(1)\\
	&\rightarrow 0 ,
\end{align*}
where the second line follows from Chebyschev's inequality, the third line by noticing that the $M_Z(1)$ is a projection matrix, which implies
\begin{align*}
	\E\max_j|T^{-1/2}\bs{z}_jM_Z(1)U|^2 & = \E\max_j|T^{-1/2}\bs{z}_j'U|^2\\
	&\le m_2^{1/d}\max_j\left(\E|T^{-1/2}\bs{z}_j'U|^{2d}\right)^{1/d}\\
	&\le m_2^{1/d}c_d.
\end{align*}

Finally, both $\Pr(\mathcal{A}^c_T)$ and $\Pr(\mathcal{B}_T^c)$ converge to $0$ and $\Pr(\mathcal{A}_T\cap\mathcal{B}_T)\rightarrow 1$, proving the theorem.

\end{proof}
\subsection{Proof of theorem \ref{thm:oracle}}
\begin{proof}

Theorem \ref{thm:modsel} tells us that the adaptive Lasso estimator \eqref{eq:adalasso} asymptotically chooses the correct set of non-zero parameters. It remains to show that the distribution of the estimator of the non-zero parameters is the same as the OLS estimator conditional on knowing the correct set of parameters. Write the derivative of the criterion function in \eqref{eq:adalasso} is given by
\begin{equation}
	Q_T(\theta) = -2(Y-W(1)\theta(1))'W(1) + 2(W(2)\theta(2))'W(1) + \lambda(1)L(1)\sgn(\theta(1)),
	\label{eq:dcrit}
\end{equation}
where $L(1)$ and $\lambda(1)$ are as in proposition \ref{p:necessary}. Setting $Q_t(\hat\theta) = 0$, and $U = Y-W(1)\theta_0(1)$, we find 
\begin{equation}
	\Gamma^{1/2}(\hat\theta(1) - \theta_0(1)) = \Omega_{11}^{-1}\Gamma^{-1/2}U'W(1) -\Omega_{11}^{-1}\left[ \Gamma^{-1/2}\hat\theta(2)W(2)'W(1) +\frac{1}{2}\Gamma^{-1/2}\lambda(1)L(1)\sgn(\theta(1))\right],
	\label{eq:estimator}
\end{equation}
which tells us the adaptive Lasso estimator has the same form of a biased OLS estimator, with the bias between square brackets. Hence, \eqref{eq:oracle} is equivalent to showing $\Omega_{11}$ converges to the optimal covariance matrix; $T^{-1}U'X(1)$ has a mixing normal distribution; $T^{-1/2}U'Z(1)$ has a normal distribution and the terms in square brackets converge to zero. 

We have already seen in proof of theorem \ref{thm:modsel} that
\[
\Omega_{11}\Rightarrow
\left(
\begin{array}{cc}
	\int B_{X(1)}B_{X(1)}'dr & \bs{0}\\
	\bs{0}'& \Sigma_{Z(1)^2}
\end{array}
\right).
\]

Since $q_1=O(1)$, it follows from assumption (DGP) that
\[
T^{-1}U'X(1)\Rightarrow	\int B_{X(1)}dB_u.
\]

Using the Cr\'amer-Wold device, one can show that for any $q_2\times 1$ vector $\alpha$ satisfying $\alpha'\alpha\le 1$, $T^{-1/2}\alpha'\E Z(1)'U=0$ and\begin{align*}
	\E\left(T^{-1/2}\alpha'Z(1)'U\right)^2 &= \alpha'\E[Z(1)'UU'Z(1)]\alpha\\
	&\rightarrow \sigma^*_{u^2}\alpha'\Sigma_{Z(1)^2}\alpha.
\end{align*}
where the last line follows from assumption (DGP). Combining the Cr\'amer-Wold device with the Central Limit theorem for dependent processes, one can show that for any constant $c$, $T^{-1/2}c\alpha'Z(1)'U\Rightarrow N(0,c^2\sigma^*_{u^2}\alpha'\Sigma^*_{Z(1)^2}\alpha)$ and therefore $T^{-1/2}Z(1)'U \Rightarrow N(\bs{0},\sigma^*_{u^2}\Sigma_{Z(1)^2}^*)$.

The first term of the bias vanishes because $\hat\theta(2)=o_p(1)$. The second term of the bias is also treated in the proof of theorem \ref{thm:modsel} and is show to be $o(\lambda_2q_2/T^{(1+\rho)/2}+\lambda_1/T^{1+\rho})$. By assumption $\lambda_1/T^{1+\rho}\rightarrow 0$ and $\lambda_2q_2/T^{(1+\rho)/2}\rightarrow 0$. Therefore the bias term converges to zero as $T$ increases.

\end{proof}

\end{document}